\documentclass[reqno, letterpaper,11pt]{amsart}
\usepackage{amsmath,enumerate,enumitem,amsfonts,color,amsthm,hyperref,bbm} 
\usepackage{appendix}
\title[Scaling Limits for Exponential Hedging]{Scaling Limits for Exponential Hedging in Trinomial Models}

\author{Yan Dolinsky} 
\address{Department of Statistics, Hebrew University}
\email{yan.dolinsky@mail.huji.ac.il}

\author{Xin Zhang} 
\address{Department of Finance and Risk Engineering, New York University}
\email{xz1662@nyu.edu}

\date{\today}

\numberwithin{equation}{section}  
\newtheorem{defn}{Definition}[section]
\newtheorem{example}[defn]{Example}

\newtheorem{thm}[defn]{Theorem}

\newtheorem{cor}[defn]{Corollary}
\newtheorem{lem}[defn]{Lemma}
\newtheorem{assume}[defn]{Assumption}

\begin{document}

\begin{abstract}
We study scaled trinomial models  converging to the Black--Scholes model, and analyze exponential certainty-equivalent prices for path-dependent European options. 
As the number of trading dates $n$ tends to infinity and the risk aversion is scaled as $nl$
for a fixed constant $l>0$, we derive a nontrivial scaling limit. Our analysis is 
purely probabilistic. Using a duality argument for the certainty equivalent, together with martingale 
and weak-convergence techniques, we show that the limiting problem takes the form of a volatility 
control problem with a specific penalty. 
For European options with Markovian payoffs, 
we analyze the optimal control problem and show that the corresponding delta-hedging strategy 
is asymptotically optimal for the primal problem.
\end{abstract}

\thanks{Y. Dolinsky is partially supported by the ISF grant 305/25. X. Zhang is partially supported by the NSF grant DMS-2508556.}

\maketitle

\section{Introduction and the Main Result}

Understanding hedging and pricing in incomplete markets is a central problem in mathematical finance. In complete models, such as the Cox--Ross--Rubinstein binomial model \cite{CRR}, a European contingent claim can be perfectly replicated, and the absence of residual risk leads to a unique price. In incomplete models, by contrast, perfect replication typically fails, and one must replace replication-based pricing by criteria that account for risk preferences. Exponential utility and the associated certainty equivalent provide a natural and tractable framework for this purpose \cite{C}.

This paper studies exponential hedging in trinomial models that converge to the Black--Scholes model. The trinomial setting is a natural extension of the binomial model: in addition to upward and downward movements, the stock may also remain unchanged. This additional flexibility destroys market completeness, but it also makes the model rich enough to capture nontrivial hedging effects in the large-\(n\) limit. Our main result identifies this limit and shows that, after a suitable scaling of risk aversion, the discrete-time exponential hedging problem converges to a continuous-time volatility control problem with an explicit entropy-type penalty. In this way, the paper provides a connection between discrete-time utility-based hedging and continuous-time stochastic control. 

The problem is also motivated by recent developments in divergence-based approaches to model uncertainty and calibration. In particular, a recent paper \cite{BZ} introduced the reciprocal specific relative entropy between continuous martingale measures through trinomial approximations. That divergence has potential applications to model calibration \cite{Av01}, and is closely related to specific relative entropy \cite{BB,BB24,BaUn23,BWZ24}  and to specific Wasserstein-type divergences \cite{BZ25}. Our scaling limit therefore gives a concrete financial interpretation of these quantities through exponential hedging in discrete incomplete markets. 

More precisely, fix a constant \(\bar{\sigma}>0\) and \(p\in(0,1)\), and consider the probability space
\[
\Omega:=\{-1,0,1\}^{\mathbb N},
\qquad
\mathbb P:=\{p/2,\,1-p,\,p/2\}^{\mathbb N}.
\]
Let \(\{\xi_i\}_{i\in\mathbb N}\) be the canonical process and let \(\mathcal F_k=\sigma\{\xi_1,\dots,\xi_k\}\), with \(\mathcal F_0\) trivial. For each \(n\), we consider the \(n\)-step trinomial model on the time grid
\[
0,\frac1n,\frac2n,\dots,1,
\]
with zero interest rate and stock price process
\[
S_k^n=S_0\prod_{j=1}^k\left(1+\frac{\bar{\sigma}}{\sqrt n}\xi_j\right),
\qquad k=0,1,\dots,n.
\]
We assume \(n\) is sufficiently large so that \(S_k^n>0\) for all \(k\). 

A trading strategy \(\gamma=(\gamma_0,\dots,\gamma_{n-1})\) is an adapted process, where \(\gamma_k\) denotes the number of shares held at time \(\frac{k}{n}\). Its terminal gain is
\[
V^\gamma=\sum_{i=0}^{n-1}\gamma_i\left(S_{i+1}^n-S_i^n\right).
\]
Let \(\Gamma_n\) denote the set of all such strategies.

Denote by $C[0,1]$ the space of continuous functions on $[0,1]$,
equipped with the uniform norm.
Denote by $C_{++}[0,1] \subset C[0,1]$ the space of all strictly positive continuous functions.
Consider a continuous map $F :C_{++}[0,1] \to \mathbb{R}$ and assume there exist constants $C, r > 0$ such that
\begin{equation}\label{3}
|F(y)| \le C \left(1 + \sup_{t \in [0,1]} \left( y_t^r + \frac{1}{y^r_t} \right) \right),
\qquad \forall\, y \in C_{++}[0,1].
\end{equation}
For any $n\in\mathbb{N}$, let $H^n:\mathbb{R}^{n+1}\to C[0,1]$ be the linear interpolation operator
\begin{align}\label{eq:Linterpolate}
H^n(z)_t:=\bigl([nt]+1-nt\bigr)z_{[nt]}
+\bigl(nt-[nt]\bigr)z_{[nt]+1},
\qquad t\in[0,1],
\end{align}
where $z=(z_0,\ldots,z_n)\in\mathbb{R}^{n+1}$ and $[t]$ denotes the integer part of $t$.

We consider a path-dependent European payoff of the form
\[
F_n:=F(H_n(S^{n})),
\]
where \(S^{n}=(S_0^n,\dots,S_n^n)\). For a fixed \(\ell>0\), we impose exponential utility with absolute risk aversion \(n\ell\), and define the corresponding \emph{certainty equivalent} by (see Chapter~2 in \cite{C})
\[
\mathcal C_n:=\frac{1}{n\ell}\log\left(\inf_{\gamma\in\Gamma_n}\mathbb E_{\mathbb P}\big[\exp(n\ell(F_n-V^\gamma))\big]\right).
\]
Economically, \(\mathcal C_n\) is the amount of certain wealth that makes the investor indifferent between selling and hedging the claim, and not entering the position. Equivalently, $\mathcal C_n$ is the unique solution of
\[
\inf_{\gamma\in\Gamma_n}
\mathbb{E}_{\mathbb{P}}
\left[\exp\left(nl\left(F_n-V^{\gamma}-\mathcal C_n\right)\right)\right]
=1.
\]

Our main theorem shows that \(\mathcal C_n\) admits a nontrivial  limit. Consider a complete probability space
$(\Omega_W,\mathcal{F}_W,\mathbb{P}_W)$ supporting
a standard Brownian motion $W$ with right-continuous filtration
$\mathcal{F}^W_t$, $t\in[0,1]$. Let $\mathcal A_{\bar{\sigma}}$ denote the set of all bounded, progressively measurable, nonnegative processes \((\alpha_t)_{t \ge 0}\) satisfying \(\alpha_t \le \bar{\sigma}\) for all \(t \ge 0\). Define
\[
S_t^\alpha
=
S_0\exp\left(\int_0^t \alpha_u\,dW_u-\frac12\int_0^t \alpha_u^2\,du\right),
\qquad t\in[0,1].
\]

\begin{thm}\label{thm.1}
We have that
\begin{align}\label{2.3}
\lim_{n\to\infty} \mathcal C_n
=
\sup_{\alpha\in\mathcal A_{\bar{\sigma}}}
\mathbb E_{\mathbb P^{\mathbb W}}\left[
F(S^\alpha)
-\frac1\ell\int_0^1
G_p\!\left(\frac{\alpha_t^2}{\bar{\sigma}^2}\right)\,dt
\right],
\end{align}
where
\[
G_p(x):=
x\log\left(\frac{x}{p}\right)
+
(1-x)\log\left(\frac{1-x}{1-p}\right),
\qquad 0\le x\le 1.
\]
\end{thm} 

This limit identifies the effective continuous-time object behind exponential hedging in scaled trinomial models. The control variable is the volatility process \(\alpha\), while the running penalty $G_p(x)$ is given by the relative entropy between the Bernoulli laws \(\mathrm{Ber}(x)\) and \(\mathrm{Ber}(p)\). Thus the limit has a clear interpretation: the investor optimizes over effective volatilities, but deviations from the reference trinomial structure are penalized entropically. 

The theorem also clarifies how the limiting problem interpolates between two familiar regimes. When \(\ell\to\infty\), the penalty disappears and the limit reduces to
\[
\sup_{\alpha\in\mathcal A_{\bar{\sigma}}}\mathbb E_{\mathbb P^{\mathbb W}}[F(S^\alpha)],
\]
which can be viewed as a \(G\)-expectation \cite{P} over volatility scenarios in \([0,\bar{\sigma}]\). This corresponds to the case where the investor aims to super-replicate the payoff $F(S)$ under any market model whose volatility lies within this uncertainty interval (see \cite{STZ}).

In the case where $\ell \to 0$, it is straightforward to check that
\[
\lim_{n \to \infty} \mathcal C_n
= \mathbb{E}_{\mathbb{P}_W}\left[
F\left(S^{\sigma}\right)\right]
\]
where $\sigma \equiv \sqrt{p}\bar \sigma$.
We recover the Black--Scholes price with volatility $\sqrt{p}\bar \sigma$.

 For fixed \(\ell\), the certainty equivalent therefore lies between classical Black--Scholes pricing $
\mathbb{E}_{\mathbb{P}_W}\left[
F\left(S^{\sigma}\right)\right]
$ and super-replication $
\sup_{\alpha \in \mathcal{A}_{\bar\sigma}}
\mathbb{E}_{\mathbb{P}_W}\left[
F\left(S^{\alpha}\right)\right]
$. This makes the limit particularly appealing: it quantifies how utility-based hedging in an incomplete discrete model interpolates between risk-neutral valuation and worst-case pricing. Moreover, if the payoff function $F$ is convex, then the above $G$-expectation coincides with the Black--Scholes price with volatility $\bar\sigma$ (see the argument in Section~6 of \cite{Ks}).

Our limit theorem can be viewed as a natural extension of the limit theorems in \cite{CD,DZ}, 
where similar scaling limits were considered for the discretized Bachelier model. 
While \cite{CD} studied European contingent claims with path-dependent payoffs, 
the proof of the upper bound was incomplete. In particular, to ensure compactness of the pre-limit, 
the authors introduced an artificial term in the discrete-time hedging strategy due to technical difficulties. 
This issue was later resolved in \cite{DZ} by applying tools from stochastic control, 
but only in the case of Markovian payoffs. 
The present trinomial structure is more amenable to a probabilistic approach, 
as martingale measures in this setting are tight.

Finally, for European options with Markovian payoffs, we analyze the limiting control problem further. We derive the associated HJB equation and show that the corresponding delta-hedging strategy is asymptotically optimal for the original discrete-time problem; see Theorem~\ref{thm:2}. Thus the paper not only identifies the scaling limit, but also extracts an effective hedging rule from it. 

The rest of the paper is organized as follows. In the next section, we establish  sensitivity properties of the certainty equivalent in trinomial models. In Section \ref{sec:4}, we prove Theorem \ref{thm.1}. Section \ref{sec:5} provides, for Markovian payoffs, an HJB formulation for the right-hand side of (\ref{2.3}) and establishes an asymptotically optimal hedging strategy for this case.

\section{Properties of the Certainty Equivalent}\label{sec:2}
This section provides a sensitivity analysis of the certainty equivalent. In particular, for convex payoffs it establishes monotonicity with respect to the model parameters $p$ and $\bar \sigma$, which helps explain how the pricing/hedging problem varies across trinomial specifications. We begin with the following technical lemma.

\begin{lem}\label{lem100}
Let $h:\mathbb{R}_+ \to \mathbb{R}$ be a convex function. 
Fix $a \in \mathbb{R}$ and $x > 0$, and define $\Upsilon :[0,1]\times\mathbb R_{+}\rightarrow\mathbb R$ by
\[
\Upsilon(\hat p,\hat{\sigma})
:= \log\left(
\frac{\hat p}{2}\left(
e^{h(x(1+\hat{\sigma})) - a \hat{\sigma} x}
+
e^{h(x(1-\hat{\sigma})) + a \hat{\sigma} x}
\right)
+
(1-\hat p)e^{h(x)}
\right).
\]
Then $\Upsilon$ is non-decreasing with respect to both variables.
\end{lem}
\begin{proof}
By applying standard density arguments, we may assume that $h$ is differentiable.
Since $h$ is convex, the function $\Delta \mapsto e^{h(x + \Delta) - a\Delta}$, $\Delta \in \mathbb R$, is also convex. Hence,
\[
\frac{\partial \Upsilon}{\partial \hat p}
=
\frac{
\frac{1}{2}\left(
e^{h(x(1+\hat{\sigma})) - a \hat{\sigma} x}
+
e^{h(x(1-\hat{\sigma})) + a \hat{\sigma} x}
\right)
-
e^{h(x)}
}{
\frac{\hat p}{2}\left(
e^{h(x(1+\hat{\sigma})) - a \hat{\sigma} x}
+
e^{h(x(1-\hat{\sigma})) + a \hat{\sigma} x}
\right)
+
(1-\hat p)e^{h(x)}
}
\geq 0
\]
and
\[
\frac{\partial \Upsilon}{\partial \hat{\sigma}}
=
\frac{
\frac{\hat p x}{2}\left(
\left(e^{h(x+\Delta)-a \Delta}\right)'_{\Delta=\hat{\sigma} x}
-
\left(e^{h(x+\Delta)-a \Delta}\right)'_{\Delta=-\hat{\sigma} x}
\right)
}{
\frac{\hat p}{2}\left(
e^{h(x(1+\hat{\sigma})) - a \hat{\sigma} x}
+
e^{h(x(1-\hat{\sigma})) + a \hat{\sigma} x}
\right)
+
(1-\hat p)e^{h(x)}
}
\geq 0.
\]
This completes the proof. 
\end{proof}

Next, fix $n\in\mathbb N$. Let $\hat F:\mathbb R^{n+1}_{+} \rightarrow\mathbb R$ be a given function. Consider the $n$-step trinomial model with payoff (of a European contingent claim) given by
\[
\hat{F}_n := \hat{F}(S^n_0, S^n_1, \ldots, S^n_n).
\]
Without loss of generality, we assume exponential utility with risk aversion normalized to~one. The corresponding certainty equivalent is given by
\[
\hat {\mathcal C}_n := \inf_{\gamma \in \Gamma_n}\log\left(
\mathbb{E}_{\mathbb{P}}
\left[
\exp\left(\hat F_n - V^{\gamma}\right)
\right]\right).
\]
Define the functions $U^{p,\bar\sigma}_k$, $k=0,1,...,n$ by backward recursion:
\begin{align}\label{200}
&U^{p,\bar\sigma}_n(z_0, \ldots, z_n):= \hat F_n(z_0, \ldots, z_n), \nonumber\\
&U^{p,\bar\sigma}_{k-1}(z_0, \ldots, z_{k-1})\\
&= \inf_{a \in \mathbb{R}} \log \Bigg(
\frac{p}{2} \Big[
\exp\Big(
U^{p,\bar\sigma}_k(z_0, \ldots, z_{k-1}, z_{k-1}(1+\bar{\sigma}/\sqrt n))
- a \bar{\sigma} z_{k-1}
\Big) \nonumber\\
&\qquad\qquad
+ \exp\Big(
U^{p,\bar\sigma}_k(z_0, \ldots, z_{k-1}, z_{k-1}(1-\bar{\sigma}/\sqrt n))
+ a \bar{\sigma} z_{k-1}
\Big)
\Big] \nonumber \\
&\qquad
+ (1-p)\, \exp\Big(
U^{p,\bar\sigma}_k(z_0, \ldots, z_{k-1}, z_{k-1})
\Big)
\Bigg). \nonumber
\end{align}
From standard dynamic programming in discrete time and space we get 
\begin{equation}\label{100}
\hat {\mathcal C}_n=U^{p,\bar\sigma}_0(S_0).
\end{equation}

For each $k$, the function $U^{p,\bar{\sigma}}_{k-1}(z_0, \ldots, z_{k-1})$ is the certainty equivalent
at a given node for a one-period trinomial model with random payoff
$$
U_k^{p,\bar\sigma}\big(z_0, \ldots, z_{k-1}, z_{k-1}(1+\bar{\sigma}\xi_k/\sqrt{n})\big).
$$
Clearly, a probability measure is a martingale measure for our trinomial model
if and only if it assigns equal probabilities to the up and down moves at each node.
Thus, from the dual representation of the certainty equivalent in a one-period model
(see Chapter~3 in \cite{FL}), we obtain
\begin{align}\label{101}
&U^{p,\bar{\sigma}}_{k-1}(z_0, \ldots, z_{k-1})
= \sup_{0<q<1} \Bigg(
\frac{q}{2} \Big[
U^{p,\bar{\sigma}}_k\big(z_0, \ldots, z_{k-1},
z_{k-1}(1+\tfrac{\bar{\sigma}}{\sqrt{n}})\big) \nonumber\\
&+ U^{p,\bar{\sigma}}_k\big(z_0, \ldots, z_{k-1},
z_{k-1}(1-\tfrac{\bar{\sigma}}{\sqrt{n}})\big)
\Big] + (1-q)\,
U^{p,\bar{\sigma}}_k\big(z_0, \ldots, z_{k-1}, z_{k-1}\big)\nonumber\\
&-q\log\left(\frac{q}{p}\right)-(1-q)\log\left(\frac{1-q}{1-p}\right)
\Bigg).
\end{align}

We arrive at the following corollary.

\begin{cor}\label{cor1}
Assume that $\hat F_n$ is convex. Then, for all $k$, the function $U^{p,\bar{\sigma}}_k$ is convex.
\end{cor}

\begin{proof}
The statement follows by backward induction using the representation \eqref{101}.
\end{proof}

\begin{lem}\label{lem.new}
Assume that $\hat F_n$ is convex. Let $1 > p_2 \geq p_1 > 0$ and $\bar{\sigma}_2 \geq \bar{\sigma}_1 > 0$. Then
$
U^{p_2,\bar{\sigma}_2}_k \geq U^{p_1,\bar{\sigma}_1}_k
$
for all $k = 0,1,\dots,n$.
\end{lem}
\begin{proof}
The proof proceeds by backward induction. For $k = n$, the statement is immediate. Assume that the statement holds for some $k$, and prove it for $k-1$.
From the induction assumption and the dynamic programming (\ref{200})
we obtain
\begin{align*}\label{201}
&U^{p_2,\bar\sigma_2}_{k-1}(z_0, \ldots, z_{k-1})\\
&\geq\inf_{a \in \mathbb{R}} \log \Bigg(
\frac{p_2}{2} \Big[
\exp\Big(
U^{p_1,\bar\sigma_1}_k(z_0, \ldots, z_{k-1}, z_{k-1}(1+\bar{\sigma}_2/\sqrt n))
- a \bar{\sigma}_2 z_{k-1}
\Big) \nonumber\\
&\qquad\qquad
+ \exp\Big(
U^{p_1,\bar\sigma_1}_k(z_0, \ldots, z_{k-1}, z_{k-1}(1-\bar{\sigma}_2/\sqrt n))
+ a \bar{\sigma}_2 z_{k-1}
\Big)
\Big] \nonumber \\
&\qquad
+ (1-p_2)\, \exp\Big(
U^{p_1,\bar\sigma_1}_k(z_0, \ldots, z_{k-1}, z_{k-1})
\Big)
\Bigg)\\
&\geq\inf_{a \in \mathbb{R}} \log \Bigg(
\frac{p_1}{2} \Big[
\exp\Big(
U^{p_1,\bar\sigma_1}_k(z_0, \ldots, z_{k-1}, z_{k-1}(1+\bar{\sigma}_1/\sqrt n))
- a \bar{\sigma}_1 z_{k-1}
\Big) \nonumber\\
&\qquad\qquad
+ \exp\Big(
U^{p_1,\bar\sigma_1}_k(z_0, \ldots, z_{k-1}, z_{k-1}(1-\bar{\sigma}_1/\sqrt n))
+ a \bar{\sigma}_1 z_{k-1}
\Big)
\Big] \nonumber \\
&\qquad
+ (1-p_1)\, \exp\Big(
U^{p_1,\bar\sigma_1}_k(z_0, \ldots, z_{k-1}, z_{k-1})
\Big)
\Bigg)
\end{align*}
where the last inequality follows from Lemma \ref{lem100} and the convexity of $U^{p_1,\bar\sigma_1}_k$.
\end{proof}

We conclude this section with the following corollary.

\begin{cor}
Assume that $\hat F_n$ is convex, and fix $n \in \mathbb{N}$. Then the certainty equivalent $\hat {\mathcal C}_n$ is non-decreasing as a function of $p$ and $\bar{\sigma}$.
\end{cor}
\begin{proof}
The result follows by combining \eqref{100} with Lemma \ref{lem.new} for $k = 0$.
\end{proof}

\section{Proof of Theorem~\ref{thm.1}}\label{sec:4}

For any $n\in\mathbb N$ denote by $\mathcal Q_n$ the set of all probability measures $\mathbb Q$ on $\mathcal F_n$ such that 
$\{S^n_k\}_{k=0}^n$ is $\mathbb Q$-martingale. Observe that $\mathbb Q_n\in\mathcal Q_n$ if and only if 
$$\mathbb Q_n(\xi_k=1|\mathcal F_{k-1})=\mathbb Q_n(\xi_k=-1|\mathcal F_{k-1}) \ \ k=1,...n.$$
\begin{lem}\label{lem2}
(i). For any $m\in\mathbb N$ we have 
\begin{equation}\label{15}
\sup_{n\in\mathbb N}\sup_{\mathbb Q_n\in\mathcal Q_n}\mathbb E_{\mathbb Q_n}\left[\left(\max_{0\leq k\leq n}S^n_k\right)^{2m}\right]<\infty, 
\end{equation}
\begin{equation}\label{eq:reciprocal}
\sup_{n\in\mathbb N}\sup_{\mathbb Q_n\in\mathcal Q_n}\mathbb E_{\mathbb Q_n}\left[\left(\max_{0\leq k\leq n}1/S^n_k\right)^{2m}\right]<\infty, 
\end{equation}
(ii). For any sequence of probability measure $\mathbb Q_n\in\mathcal Q_n$, $n\in\mathbb N$ the sequence of distributions 
$(H^n(S^n);\mathbb Q_n)$, $n\in\mathbb N$ is tight in the space $C[0,1]$. 
\end{lem}
\begin{proof}
\emph{(i)}. Choose $m\in \mathbb Z$.
Fix $n$ and $\mathbb Q_n\in\mathcal Q_n$. From the 
equality $\mathbb E_{\mathbb Q_n}[\xi_k | \mathcal{F}_{k-1}]=0$ it follows that 
 there exists a constant $c=c(m)$ (depends only on $m$) such that 
\begin{align*}
\mathbb E_{\mathbb Q_n}\left[(S^n_k)^{2m}\left|\right.\mathcal F_{k-1}\right]&= (S^n_{k-1})^{2m}  \, \sum_{i=0}^{2m} \binom{2m}{i} \,\mathbb E_{\mathbb Q_n}[(\bar{\sigma}\xi_k/\sqrt{n})^{i}\left|\right.\mathcal F_{k-1}] \\
&\leq (S^n_{k-1})^{2m} (1+c/n).
\end{align*}
This together with the Doob's maximal inequality for the sub-martingale 
$(S^n_k)^{2m}$, $k=0,1...,n$ 
gives \eqref{15}.

The proof of \eqref{eq:reciprocal} is similar. Since $S_k^n$ is positive and the function $\mathbb{R}_{+} \ni x \mapsto 1/x$ is convex, it follows from Jensen's inequality that $(1/S_k^n)_{k=0}^n$ is a positive submartingale. 
Moreover, it is easy to verify that for any $x \in (-1/2, 1/2)$,
$
\frac{1}{1+x} \leq 1 - x + 2x^2.
$
Hence, by applying again the identity
$
\mathbb{E}_{\mathbb{Q}_n}[\xi_k \mid \mathcal{F}_{k-1}] = 0,
$
we obtain (for sufficiently large $n$) that
\begin{align*}
\mathbb E_{\mathbb Q_n}\left[(1/S^n_{k})^{2m}\left|\right.\mathcal F_{k-1}\right]&= (1/S^n_{k-1})^{2m} \mathbb E_{\mathbb Q_n}[(1-\bar{\sigma}\xi_k/\sqrt{n}+2\bar{\sigma}^2\xi_k^2/n)^{2m}\left|\right.\mathcal F_{k-1}] \\
&\leq (1/S^n_{k-1})^{2m} (1+\hat c/n)
\end{align*}
for some $\hat c$ only depends on $m$. Then we can conclude as in the proof for \eqref{15}. \\

\noindent\emph{(ii)}. We apply the Kolmogorov tightness criterion (for details see \cite{B}). Namely, in order to prove the statement it is sufficient to establish that there 
exists a constant $C>0$ such that 
\begin{equation}\label{16}
\mathbb E_{\mathbb Q_n}\left[\left(H^n_t(S^n)-H_s^n(S^n)\right)^4\right]\leq C(t-s)^2, \ \  n\in\mathbb N, \ 0\leq s<t\leq 1.
\end{equation}
To this end fix $n\in\mathbb N$ and  $0\leq s<t\leq 1$. If $t-s\leq 1/n$
then there exists a constant $c_1>0$ such that 
$$|H^n_t(S^n)-H_s^n(S^n)|^4\leq c_1n^2(t-s)^4\left(\max_{0\leq k\leq n}S^n_k\right)^4 \leq c_1(t-s)^2\left(\max_{0\leq k\leq n}S^n_k\right)^4,$$
and thus 
\eqref{16} follows from \eqref{15}. 
Thus, assume that $t-s>1/n$.
Set $i=[ns]+1$ and $j=[nt]$.

From the  Burkholder-Davis-Gundy inequality
we obtain that there exists a constant $c_2>0$ such that 
\begin{align*}
&\mathbb E_{\mathbb Q_n}\left[\left(S^n_j-S^n_i\right)^4\right]\leq c_2\mathbb E_{\mathbb Q_n}\left[\left(\sum_{k=i}^{j-1} (S^n_{k+1}-S^n_k)^2\right)^2\right]\\
&\leq c_2\mathbb E_{\mathbb Q_n}\left[\left(\frac{ \bar\sigma^2 (j-i)}{n}(\max_{0\leq k\leq n}S^n_k)^2\right)^2\right]\\
&\leq c_2\bar\sigma^4(t-s)^2 \mathbb E_{\mathbb Q_n}\left[\left(\max_{0\leq k\leq n}S^n_k\right)^4\right].
\end{align*}
This together with the inequality (recall that $t-s>1/n$)
$$|H^n_t(S^n)-H_s^n(S^n)|\leq |S^n_j-S^n_i|+2\frac{\bar\sigma}{\sqrt n}\max_{0\leq k\leq n}S^n_k$$
gives \eqref{16}. 
\end{proof}

Next, Let $D[0,1]$ be the space of all càdlàg (right continuous with left limits) functions equipped with the Skorohod topology (see~\cite{B}). For any $n \in \mathbb N$ and $ z \in \mathbb R^{n+1}$, define 
$$L^n(z):=\left\{z_{[nt]} \right\}_{t=0}^1 \in D[0,1].$$
For any $n\in\mathbb N$ and $\mathbb Q_n\in\mathcal Q_n$, define
\begin{align*}
&\Phi^{n,\mathbb Q_n}_k:=\bar\sigma^2\mathbb E_{\mathbb Q_n}[\xi^2_k|\mathcal F_{k-1}], \ \ k=1,...,n\\
&\Psi^{n,\mathbb Q_n}_k=\frac{1}{n}\sum_{i=1}^k \Phi^{n,\mathbb Q_n}_i, \ \ k=0,1,...,n.
\end{align*}
\begin{lem}\label{lem3}
Any subsequence (still denote it by $n$ for simplicity) of 
probability measures $\mathbb Q_n\in\mathcal Q_n$, $n\in\mathbb N$ contains a further subsequence 
along which we have the weak convergence 
$$(L^n(S^n);L^n(\Psi^{n,\mathbb Q_n});\mathbb Q_n)\Rightarrow (\{M_t\}_{t=0}^1,\{\langle  \log M\rangle_t \}_{t=0}^1;\mathbb P^M) \ \mbox{on} \ D^2[0,1]$$
for some continuous and strictly positive martingale $M=(M_t)_{0\le t\le 1}$
on a suitable probability space $(\Omega^M,\mathcal{F}^M,P^M)$. As usual $\langle\cdot\rangle$ denotes the quadratic variation of 
the continuous semi-martingale $\cdot$.  
\end{lem}
\begin{proof}
By Lemma~\ref{lem2}, $(H^n(S^n); \mathbb Q_n)$ is tight in $C[0,1]$. Therefore, there exists a weak limit $(\{M_t\}_{t=0}^1;\mathbb P^M)$ defined on a probability space $(\Omega^M,\mathcal{F}^M,P^M)$, which is also a weak limit of $(L^n(S^n); \mathbb Q_n)$ in $D[0,1]$. As $(L^n(S_n);\mathbb Q_n)$ is a martingale, together with the uniform integrability of marginal distributions, the weak limit is  a non-negative martingale.

Let us prove that $\{M_t\}_{t=0}^1$ is strictly positive. Define a function $I$ on $D[0,1]$ via $I(\omega):=\inf_{t \in [0,1]}  \omega_t$. It can be shown that $I$ is continuous under the Skorokhod topology. Hence, from the Fatou lemma and \eqref{eq:reciprocal} it follows that $M$ is strictly positive. 
Therefore $$(L^n(S^n),L^n(1/S^n); \mathbb Q_n) \Rightarrow (M,1/M; \mathbb P^M).$$

Take $N_t:=\int_{0}^t\frac{dM_u}{M_u}$, $t\in [0,1]$, and take
\begin{align*}
    N^n_k:= \int_0^{k/n} \frac{dL^n(S^n)_t}{L^n(S^n)_t}=\frac{\bar\sigma}{\sqrt{n}}\sum_{i=1}^k \xi_i, \quad k=1, \dotso, n.
\end{align*}
Invoking \cite[Theorem 4.3]{DP},  we obtain the weak convergence $$\left( L^n(S^n),L^n(N^n); \, \mathbb Q_n\right) \Rightarrow \left(M,N ; \mathbb P^{M} \right).$$ 

Clearly, $\{N^n_k\}_{k=1}^n$ is a $\mathbb Q_n$-martingale with uniformly bounded jumps. Define its quadratic variation process 
$\Lambda^n_k=\frac{\bar\sigma^2}{n} \sum_{i=1}^k \xi^2_i$, $k=0,1,...,n$. Hence, from Theorem 4.3 in \cite {DP} and the equalities 
\begin{align*}
&\Lambda^n_k=|N^n_k|^2-2\sum_{i=1}^k N^n_{i-1}(N^n_i-N^n_{i-1}), \,  k=1,...,n,\\
&\langle  \log M\rangle_t=\langle N\rangle_t=N^2_t-2\int_{0}^t N_u dN_u, \ t\in [0,1] 
\end{align*}
we obtain 
$$\left(L^n(S^n), L^n(\Lambda^n);\mathbb Q_n\right)\Rightarrow 
(M,  \langle\log M\rangle ;\mathbb P^M).$$

Thus, in order to complete the proof it remains to show that 
\begin{equation}\label{10+}
\lim_{n\rightarrow\infty}\max_{0\leq k\leq n}|\Lambda^n_k-\Psi^{n,\mathbb Q_n}_k|= 0 \ \mbox{in} \ \mbox{probability}.
\end{equation}
To this end we observe that for any $n$ the process $\Lambda^n_k-\Psi^{n,\mathbb Q_n}_k$, $0\leq k\leq n$ is a martingale and the corresponding quadratic variation satisfies 
 $$\sum_{k=1}^n \left(\Lambda^n_k-\Psi^{n,\mathbb Q_n}_k-(\Lambda^n_{k-1}-\Psi^{n,\mathbb Q_n}_{k-1})\right)^2\leq  4\bar{\sigma}^2/n.$$
This together with the Burkholder-Davis-Gundy inequality gives \eqref{10+}. 
\end{proof}

Next, let $(W_t)_{[0,1]}$ be the canonical process on $D[0,1]$, and let $\mathcal A^{0}_{\bar\sigma}\subset\mathcal A_{\bar\sigma}$
be the subset of all $\alpha\in  A_{\bar\sigma}$ of the form
\begin{equation} \label{7}
\alpha_t:= \phi_0\mathbb I_{[t_0,t_1)}+\sum_{j=1}^{J-1} \phi_{j-1}\left(W_{t_0},\dots,W_{t_{j-1}}\right) \mathbb I_{t\in [t_{j},t_{j+1})}, \quad t\in [0,T], 
\end{equation} 
where $0=t_0<t_1<\dots<t_J=T$ is a finite deterministic partition of $[0,T]$, $\phi_0>0$ is a constant and for each $j=1,...,J-1$,
$\phi_j:\mathbb R^{j}\rightarrow (\epsilon,\bar\sigma)$, (for some $\epsilon>0$) is continuous. Without loss of generality 
we assume that $\phi_0>\epsilon$. 
\begin{lem}\label{lem1}
Let $\alpha\in\mathcal A^0_{\bar\sigma}$. There exists a sequence of probability measures 
$\mathbb Q_n\in\mathcal Q_n$, $n\in\mathbb N$ such that we have the weak convergence 
\begin{equation}\label{9}
(L^n(S^n),L^n(\Phi^{n,\mathbb Q_n});\mathbb Q_n)\Rightarrow (\{S^{\alpha}_t\}_{t=0}^1,\{ \alpha^2_t\}_{t=0}^1;\mathbb P_W) \ \mbox{on} \ D^2[0,1]
\end{equation}
where $\Phi^n_0:=\phi_0$.
\end{lem}
\begin{proof}
Let $\alpha$ be as in \eqref{7}. 
Fix $n\in\mathbb N$.
For $k=1,\dots,n$, let $j:=j(k)$ be such that $k/n\in[t_j,t_{j+1})$.
Define 
$W^n_0:=0$, $\Phi^n_0:=\phi_0$ and by recursion for $k=1,...,n$
\begin{align*}
&\Phi^n_k:=
\phi^2_{j(k)-1}\left(
W^n_{[nt_0]},\dots,W^n_{[nt_{j(k)-1}]}
\right),\\
&W^n_k:=W^n_{k-1}
+\frac{\bar\sigma \xi_k}{\sqrt{n\Phi^n_k }}.
\end{align*}
Let us notice that there exists a unique probability measure $\mathbb Q_n\in\mathcal Q_n$ such that 
$\Phi^{n,\mathbb Q_n}_k=\Phi^n_k$ for all $k=1,...n$. 

Next, observe that \(\{W^n_k\}_{k=0}^{n}\) is a $\mathbb Q_n$- martingale and for any \(k=1,\dots,n\)
\[
|W^n_k-W^n_{k-1}|\leq \frac{1}{\varepsilon\sqrt{n}},
\qquad
\mathbb E_{\mathbb Q_n}\big((W^n_k-W^n_{k-1})^2\mid\mathcal F_{k-1}\big)=\frac{1}{n}.
\]
Thus, by the Martingale Central Limit Theorem (Theorem 7.4.1 in \cite{EK}) we obtain
\[
\big(\{W^n_{[nt]}\}_{t=0}^1;\mathbb Q_n\big)\ \Rightarrow\ \big(\{W_t\}_{t=0}^1;\mathbb P_W\big)
\quad\text{on }D[0,1].
\]
From the continuity of the functions \(\phi_1,\dots,\phi_J\) we obtain the joint convergence
\begin{equation}\label{10}
\big(\{W^n_{[nt]}\}_{t=0}^1,\{\Phi^{n,\mathbb Q_n}_{[nt]}\}_{t=0}^1;\mathbb Q_n\big)
\ \Rightarrow\ 
\big(\{W_t\}_{t=0}^1,\{\alpha^2_t\}_{t=0}^1;\mathbb P_W\big)
\ \mbox{on} \ D^2[0,1].
\end{equation}
By combining Theorems 4.3-4.4 in \cite{DP}, \eqref{10} and the relations 
\begin{align*}
&S^n_0=S_0, \ S^n_k-S^n_{k-1}=\sqrt{\Phi^{n,\mathbb Q_n}_k}S^n_{k-1} (W^n_k-W^n_{k-1}), \ \ k=1,...,n\\
&S^{\alpha}_0=S_0, \ dS^{\alpha}_t=\alpha_t S^{\alpha}_t dW_t, \ \ t\in [0,1],
\end{align*}
we obtain \eqref{9}.
\end{proof}

Next, we apply the dual representation for the certainty equivalent in discrete time (see Chapter 11 in \cite{FL})
which says 
\begin{equation}\label{4.0}
\mathcal C_n:=\sup_{\mathbb Q\in\mathcal Q_n}\mathbb E_{\mathbb Q}\left[F_n-\frac{1}{n l}\log\left(\frac{d\mathbb Q}{d\mathbb P}\right)\right].
\end{equation} 
Fix $n$ and $\mathbb Q_n\in\mathcal Q_n$. Recall the process $\Phi^{n,\mathbb Q_n}$ which is given by before Lemma \ref{lem3}.
Observe that for any $k=1,...,n$
\begin{align*}
&\mathbb Q_n(\xi_k=1|\mathcal F_{k-1})=\mathbb Q_n(\xi_k=-1|\mathcal F_{k-1})=\frac{\Phi^{n,\mathbb Q_n}_k}{2\bar\sigma^2},\\
&\mathbb Q_n(\xi_k=0|\mathcal F_{k-1})=1-\frac{\Phi^{n,\mathbb Q_n}_k}{\bar\sigma^2}.
\end{align*}
Hence,
\begin{align*}
&\mathbb E_{\mathbb Q_n}\left[\log\left(\frac{d\mathbb Q_n}{d\mathbb P}\right)\right]\\
&=\sum_{k=1}^n \mathbb E_{\mathbb Q_n}\left[\log\left(\frac{\mathbb Q_n(\xi_k|\mathcal F_{k-1})}{\mathbb P(\xi_k|\mathcal F_{k-1})}\right)\right]\\
&=\sum_{k=1}^n \mathbb E_{\mathbb Q_n}\left[\mathbb E_{\mathbb Q_n}\left[\log\left(\frac{\mathbb Q_n(\xi_k|\mathcal F_{k-1})}{\mathbb P(\xi_k|\mathcal F_{k-1})}\right)\left|\mathcal F_{k-1}\right.\right]\right]\\
&=\sum_{k=1}^n \mathbb E_{\mathbb Q_n}\left[\frac{\Phi^{n,\mathbb Q_n}_k}{\bar\sigma^2}\log\left(\frac{ \Phi^{n,\mathbb Q_n}_k}{p\bar\sigma^2}\right)+
\left(1-\frac{\Phi^{n,\mathbb Q_n}_k}{\bar\sigma^2}\right)\log\left(\frac{1-\frac{\Phi^{n,\mathbb Q_n}_k}{\bar\sigma^2}}{1-p}\right)\right].
\end{align*}
Thus,
\begin{equation}\label{4.2+}
\mathbb E_{\mathbb Q_n}\left[\log\left(\frac{d\mathbb Q_n}{d\mathbb P}\right)\right]=\sum_{k=1}^n \mathbb E_{\mathbb Q_n}\left[G_{p}\left(\frac{\Phi^{n,\mathbb Q_n}_k}{\bar\sigma^2}\right)\right].
\end{equation}

Now we have all the ingredients to prove Theorem \ref{thm.1}.
\begin{proof}
\emph{Lower bound:}
First, we show that
$$
\lim\inf_{n\to\infty} \mathcal C_n\geq
\sup_{\alpha\in\mathcal{A}_{\bar\sigma}}
\mathbb{E}_{\mathbb{P}_W}\left[
F\left(S^{\alpha}\right)
-\frac{1}{l}\int_{0}^1 G_{p}\left(\frac{\alpha_t^2}{\bar\sigma^2}\right)\,dt
\right].
$$
By applying standard density arguments as in Lemma 7.3 in \cite{DS} we see that 
\begin{align*}
&\sup_{\alpha\in\mathcal{A}_{\bar\sigma}}
\mathbb{E}_{\mathbb{P}_W}\left[
F\left(S^{\alpha}\right)
-\frac{1}{l}\int_{0}^1 G_{p}\left(\frac{\alpha_t^2}{\bar\sigma^2}\right)\,dt
\right]\\
&=\sup_{\alpha\in\mathcal{A}^0_{\bar\sigma}}
\mathbb{E}_{\mathbb{P}_W}\left[
F\left(S^{\alpha}\right)
-\frac{1}{l}\int_{0}^1 G_{p}\left(\frac{\alpha_t^2}{\bar\sigma^2}\right)\,dt
\right].
\end{align*}
Thus, in order to establish the lower bound we need to prove that for any $\alpha\in\mathcal A^0_{\bar\sigma}$
\begin{equation}\label{4.1}
\lim\inf_{n\to\infty} \mathcal C_n\geq
\mathbb{E}_{\mathbb{P}_W}\left[
F\left(S^{\alpha}\right)
-\frac{1}{l}\int_{0}^1 G_{p}\left(\frac{\alpha_t^2}{\bar\sigma^2}\right)\,dt
\right].
\end{equation}
Indeed, let $\alpha$ be as in \eqref{7}. From Lemma \ref{lem1} there exists a sequence 
$\mathbb Q_n\in\mathcal Q_n$, $n\in\mathbb N$ such that \eqref{9} holds true. 
From Lemma \ref{lem2}, Lemma \ref{lem1}, the continuity of $F$ and the growth condition \eqref{3}, we obtain
\begin{equation}\label{4.2}
\lim_{n\rightarrow\infty}\mathbb E_{\mathbb Q_n}[F_n]=\mathbb{E}_{\mathbb{P}_W}\left[F\left(S^{\alpha}\right)\right].
\end{equation}
Next, since the function $G_{p}$ is bounded then Lemma \ref{lem1} and \eqref{4.2+} yield
$$\lim_{n\rightarrow\infty}\mathbb E_{\mathbb Q_n}\left[\frac{1}{n}\log\left(\frac{d\mathbb Q_n}{d\mathbb P}\right)\right]
=\mathbb{E}_{\mathbb{P}_W}\left[\int_{0}^1 G_{p}\left(\frac{\alpha^2_t}{\bar\sigma^2}\right)dt.\right]$$
This together with \eqref{4.0} and \eqref{4.2} gives \eqref{4.1}, and completes the proof of the lower bound. \\

\smallskip

\noindent\emph{Upper bound:}
By passing to a subsequence we assume without loss of generality 
that 
$\lim_{n\rightarrow\infty}\mathcal C_n$ exists. 
In view of (\ref{4.0}) there exists 
$\mathbb Q_n\in\mathcal Q_n$, $n\in\mathbb N$ (with abuse of notations) such that 
\begin{equation}\label{4.3}
\mathcal C_n<\frac{1}{n}+\mathbb E_{\mathbb Q_n}\left[F_n-\frac{1}{n l}\log\left(\frac{d\mathbb Q_n}{d\mathbb P}\right)\right], \ \ \forall n\in\mathbb N.
\end{equation}
From the Skorohod representation theorem (Theorem 3 of \cite{D}), Lemma \ref{lem3}, and passing to a subsequence
it follows that there exists a probability space
$(\tilde\Omega,\tilde{\mathcal{F}},\tilde{\mathbb P})$
on which 
\begin{equation}\label{4.5}
(L^n(S^n);L^n(\Psi^{n,\mathbb Q_n});\mathbb Q_n)\rightarrow (M,\langle  \log M\rangle;\mathbb P^M) \ \ \mbox{a.s.} \ \mbox{on} \ D^2[0,1],
\end{equation}
where $M$ is a continuous and strictly positive martingale.

Again, Lemma \ref{lem2}, the continuity of $F$ and the growth condition \eqref{3} implies 
\begin{equation}\label{4.4}
\lim_{n\rightarrow\infty}\mathbb E_{\mathbb Q_n}[F_n]=\mathbb{E}_{\tilde{\mathbb P}}\left[F\left(M\right)\right].
\end{equation}

Next, (observe that $\Phi^{n,\mathbb Q_n}\in [0,\bar\sigma^2]$) 
by Lemma A1.1 in \cite{DS1}, we construct a sequence
$$
\{\eta^n_t\}_{t=0}^1 \in \mathrm{conv}\left(\{\Phi^{n,\mathbb Q_n}_{[nt]}\}_{t=0}^1,\{\Phi^{n+1,\mathbb Q_{n+1}}_{[(n+1)t]}\}_{t=0}^1,\dots\right), \ n\in\mathbb N$$
such that $\eta^n_t$ converges almost surely in
$dt\otimes \tilde{\mathbb P}$ to a stochastic process $\eta$.
Clearly, $\eta\in  [0,\bar\sigma^2]$.

From the Lebesgue dominated convergence theorem and \eqref{4.5}
we obtain
$$
\int_0^t \eta_u\,du
=
\lim_{n\to\infty}
\int_0^t \eta^n_u\,du
=
\lim_{n\to\infty}
\int_0^t \Phi^{n,\mathbb Q_n}_{[nu]}\,du
=\langle  \log M\rangle_t, \ dt\otimes\tilde{\mathbb P}-\mbox{a.s.}
$$
Hence,
$$\frac{d}{dt}\langle  \log M\rangle_t=\eta_t, \  dt\otimes\tilde{\mathbb P}-\mbox{a.s.} $$

Observe that $G_{p}:[0,1]\rightarrow \mathbb R_{+}$ is convex 
and so 
\begin{align*}
&\mathbb E_{\tilde{\mathbb P}}\left[\int_{0}^1 G_{p}\left(\frac{\frac{d}{dt}\langle  \log M\rangle_t}{\bar\sigma^2}\right)dt\right]\\
&=\lim_{n\rightarrow\infty}\mathbb E_{\tilde{\mathbb P}}\left[\int_{0}^1 G_{p}\left(\frac{\eta^n_t}{\bar\sigma^2}\right)dt\right]\\
&\leq \lim\sup_{n\rightarrow\infty} \mathbb E_{\mathbb Q_n}\left[\int_{0}^1 G_{p}\left(\frac{\Phi^{n,\mathbb Q_n}_{[nt]}}{\bar\sigma^2}\right)dt\right].
\end{align*}
This together with \eqref{4.3} and \eqref{4.4} gives 
$$\lim_{n\rightarrow\infty}\mathcal C_n\leq\mathbb E_{\tilde{\mathbb P}}\left[F(M)-\frac{1}{l}\int_{0}^1G_{p}\left(\frac{\frac{d}{dt}\langle  \log M\rangle_t}{\bar\sigma^2}\right)dt\right].$$
Finally, since $\frac{d}{dt}\langle  \log M\rangle_t\in [0,\bar\sigma^2]$, then by applying the randomization technique as in Lemma 7.2 in \cite{DS} 
we obtain 
\begin{align*}
&\mathbb{E}_{\tilde{\mathbb{P}}}\left[
F(M)
-\frac{1}{l}\int_{0}^{1}
G_{p}\left(
\frac{\frac{d}{dt}\langle \log M\rangle_t}{\bar{\sigma}^2}
\right)\,dt
\right]\\
&\leq \sup_{\alpha\in\mathcal{A}_{\bar{\sigma}}}
\mathbb{E}_{\mathbb{P}_W}\left[
F\left(S^{\alpha}\right)
-\frac{1}{l}\int_{0}^{1}
G_{p}\left(\frac{\alpha_t^2}{\bar{\sigma}^2}\right)\,dt
\right].
\end{align*}
and complete the proof. 
\end{proof}

\section{Markovian Case}\label{sec:5}
In this section, we assume Markovian payoffs, i.e. 
$F_n=F(S^n_n)$, $n\in\mathbb N$,  for some $F \in C(\mathbb R_+; \mathbb R)$.
As in \eqref{200}, iteratively, we can define for $i=n,n-1,\dotso, 0$,
\begin{align*}
U^n_n(x)=&F(x), \\
    nl U^n_{i}(x)=& \inf_{a \in \mathbb R} \log \mathbb E_{\mathbb P}\left[\exp\left(nl U^n_{i+1}(S_{i+1}^n)-nla(S_{i+1}^n-S_i^n)  \right)\,| \, S_i^n=x\right] .
\end{align*}
Assuming that $U^n$ converges to a function $v(t,x)$, then from the equality above we deduce that  for large enough $n$
\begin{align*}
    1= \inf_{a \in \mathbb R}\mathbb E &\left[ \exp\left(lv_t(t,x)+nl  v_x(t,x) (S_{i+1}^n-x) \right. \right.\\
    &\left.\left. \ \ +nl v_{xx}(t,x) (S_{i+1}^n-x)^2/2- nl a (S_{i+1}^n-x) \right) \big| \, S_i^n=x\right].
\end{align*}
Heuristically we set $a= v_x(t,x)$, and get the PDE
\begin{equation}\label{eq:backward-v}
\begin{cases}
v_t(t,x)+K\bigl(x^2v_{xx}(t,x)\bigr)=0,
& (t,x)\in[0,1)\times(0,\infty),\\[0.3em]
v(1,x)=F(x), & x>0,
\end{cases}
\end{equation}
where
\[
\Lambda :=\bar\sigma^2>0,
\qquad
K(w):=\frac{1}{l}\log\left((1-p)+p \exp\left(\frac{l}{2}\Lambda w \right)\right),
\qquad w\in\mathbb R.
\]
Equation \eqref{eq:backward-v} is exactly the HJB  corresponding to the right hand side of \eqref{2.3}. 

In the calculation above, the choice of $a=v_x(t,x)$ gives a trading strategy in the primal problem. In Theorem~\ref{thm:2}, we will prove that such a strategy asymptotically optimal.

As \(\Lambda>0\), the function \(K\) is smooth, strictly increasing, and strictly convex on \(\mathbb R\), with
\begin{align}\label{eq:Kbound}
K'(w) \in (0,\Lambda/2), \quad K''(w) \in (0, l \Lambda^2/16), \quad w \in \mathbb R.
\end{align}
Therefore, \eqref{eq:backward-v} is locally uniformly elliptic and convex.

We transform \eqref{eq:backward-v} to forward parabolic equations. Setting $u(t,x):=v(1-t,x)$, 
\begin{equation}\label{eq:forward-u}
\begin{cases}
u_t(t,x)=K\bigl(x^2  u_{xx}(t,x)\bigr),
& (t,x)\in(0,1]\times (0,\infty),\\[0.3em]
u(0,x)=F(x), & x\in (0,\infty). 
\end{cases}
\end{equation}
Take $f(y):=F(e^y)$, $y \in \mathbb R$. By a log-space transformation, 
\begin{align*}
    \bar u(t,y):= u(t, e^y), \quad (t,y) \in [0,1] \times \mathbb R, 
\end{align*}
satisfies the forward Cauchy problem
\begin{equation}\label{eq:forward-baru}
\begin{cases}
\bar u_t(t,y)=K\bigl(\bar u_{yy}(t,y)-\bar u_y(t,y)\bigr),
& (t,y)\in(0,1]\times\mathbb R,\\[0.3em]
\bar u(0,y)=f(y), & y\in\mathbb R,
\end{cases}
\end{equation}
where we use equalities
\[
\bar u_y(t,y)=x\,u_x(t,e^y),
\qquad
\bar u_{yy}(t,y)-\bar u_y(t,y)=x^2 u_{xx}(t,e^y).
\]
Clearly, \eqref{eq:forward-u} and \eqref{eq:forward-baru} are equivalent. Throughout the remainder of this paper, we shall pass from one formulation to the other whenever convenient for the analysis. Under some assumptions on $f(\cdot)=F(e^{\cdot})$, we will prove that there exists a unique classical solution to \eqref{eq:forward-u}. 

\begin{assume}\label{assume:f} 
  Take $f(y):= F(e^y)$, $y \in \mathbb R$. 
    \begin{itemize}
        \item[(i)] $f$ is continuous with at most linear growth. \vspace{0.3em}
        \item[(ii)] $f$ is second-order differentiable and locally $\alpha$-H\"{o}lder continuous for some $\alpha \in (0,1)$, i.e. $f \in C^{2,\alpha}_{\mathrm{loc}}(\mathbb R)$. \vspace{0.3em}
        \item[(iii)] With $w_0:=f''-f'$, it holds that  
        \begin{align*}
-\infty <m_0:=&\inf_{y\in\mathbb R}w_0(y) \leq \sup_{y\in\mathbb R}w_0(y)=:M_0<\infty.
\end{align*}
\item[(iv)] $w_0 \in C^{2+\alpha}(\mathbb R)$ for some $\alpha \in (0,1)$. 
    \end{itemize}
\end{assume}

\begin{lem}\label{lem:comparison}
    Under Assumption~\ref{assume:f}~(i), a comparison principle holds for \eqref{eq:forward-baru} among the functions with at most linear growth. In particular, there exists a unique viscosity solution with at most linear growth to \eqref{eq:forward-baru}.
\end{lem}
\begin{proof}
For the existence,  a viscosity solution  to \eqref{eq:forward-baru} is provided by the stochastic control problem corresponding to the right hand side of \eqref{2.3}.

    As $r \mapsto K(r)$ is strictly increasing, equation \eqref{eq:forward-baru} is parabolic. Noting that $K$ is Lipschitz, comparison principle holds for functions with at most linear growth by a standard doubling variable argument. 
\end{proof}

\begin{lem}\label{lem:smooth}
Under Assumption~\ref{assume:f} (i) (ii), there exists a unique classical solution to \eqref{eq:forward-u}. Moreover, we have the regularity $$ u \in C^{1+\alpha/2,\,2+\alpha}_{\mathrm{loc}}([0,1] \times (0,\infty))  \cap  C^{\infty}((0,1] \times (0,\infty)).$$ 
\end{lem}
\begin{proof}
\emph{Step 1: H\"{o}lder regularity.}
According to  \cite[Theorem 2.3]{Goffi2024}, for any fixed $x_0 >0$, the solution to 
\begin{align*}
     \phi_t(t,x)= K(x_0^2 \phi_{xx}(t,x))
\end{align*}
belongs to $C^{1+\beta/2,\,2+\beta}_{\mathrm{loc}}((0,1] \times (0,\infty))$ for a universal small $\beta>0$.
Applying \cite[Theorem 1.1]{Wang1992II}, the solution to \eqref{eq:forward-u} belongs to $C^{1+\alpha/2,\,2+\alpha}_{\mathrm{loc}}((0,1] \times (0,\infty))$ with some positive $\alpha < \beta$. 

Note that $f \in C^{2,\alpha}_{\mathrm{loc}}(\mathbb R)$ is equivalent to $F \in C^{2,\alpha}_{\mathrm{loc}}(0,\infty)$. Applying \cite[Theorem 2.14]{Wang1992II}, we get the regularity at $t=0$, and conclude that $ u \in C^{1+\alpha/2,\,2+\alpha}_{\mathrm{loc}}([0,1] \times (0,\infty))$.

\medskip
\noindent\emph{Step 2: Schauder bootstrapping.} Recall that \eqref{eq:forward-u} and \eqref{eq:forward-baru} are equivalent. Differentiate \eqref{eq:forward-baru}
w.r.t time, and setting $\phi= \bar u_t$, we get that 
\begin{align*}
     \phi_t(t,y)= K' (\bar u_{yy}(t,y)-\bar u_y(t,y) ) \left(\phi_{yy}(t,y)- \phi_y(t,y)\right).
\end{align*}
Restricted to compact domains $[0,1]\times [-R,R]$, $a(t,y):=K'(\bar u_{yy}(t,y)-\bar u_y(t,y))$ is $\alpha$-H\"{o}lder continuous, and $a(t,y) \in [h, H]$ with some $0< h <H$. Hence, the above equation is uniform elliptic. As $\phi=  \bar  u_t(t,y)$ is also H\"{o}lder continuous on the parabolic boundary of $[0,1]\times [-R,R]$, using the regularity results for linear parabolic equations, e.g. \cite[Theorem 6.48]{Lieberman1996}, 
\begin{align*}
    \phi \in C^{1+\alpha/2,\,2+\alpha}_{\mathrm{loc}}((0,1] \times \mathbb R).
\end{align*}
Repeating this step, we get that $ \bar u \in C^{\infty}((0,1] \times \mathbb R )$.
\end{proof}

\begin{lem}\label{lem:PDEbound}
Under Assumption~\ref{assume:f} (i)-(iii), the solution $\bar u$ to \eqref{eq:forward-baru} satisfies 
\begin{align*}
   \bar u_t  \in [K(m_0),K(M_0)], \quad \bar u_{yy}-\bar u_y \in [m_0,M_0], \quad (t,y) \in [0,1]\times \mathbb R.
\end{align*}
\end{lem}
\begin{proof}
 For any continuous functions $g:\mathbb R \to \mathbb R$ with at most linear growth, define $\mathcal{S}_t (g)$ to be the solution to \eqref{eq:forward-baru} with initial condition $g$. It can be easily verified  that
\begin{enumerate}
    \item[(i)] $\mathcal{S}_t (\mathcal{S}_h (g))=\mathcal{S}_{t+h}(g)$. 
    \item[(ii)] $\mathcal{S}_t(g+c)=\mathcal{S}_t(g)+c$, where $c$ is a constant.
    \item[(iii)] $\mathcal{S}_t(g)\leq \mathcal{S}_t(l)$ if $g \leq l$ and with at most linear growth. 
\end{enumerate}

Let us define 
\begin{align*}
    \phi(t,y):= f(y)+t K(m_0),\\
    \psi(t,y):= f(y)+t K(M_0).
\end{align*}
As $K$ is increasing, $\phi$ and $\psi$ are subsolution and supersolution to \eqref{eq:forward-baru} respectively. 
Therefore by comparison principle in Lemma~\ref{lem:comparison}, 
$$\phi(t,y) \leq \bar u(t,y) \leq \psi(t,y).$$
Applying $\mathcal{S}_h$, $h>0$, to the inequality above, we obtain 
\begin{align*}
    \bar u(t,y)+hK(m_0) \leq \bar u(t+h,y) \leq \bar u(t,y)+hK(M_0),
\end{align*}
and hence 
\begin{align*}
K(m_0) \leq \frac{\bar u(t+h,y)-\bar u(t,y)}{h} \leq K(M_0).
\end{align*}
Letting $h \downarrow 0$, we conclude the result. 

\end{proof}

\begin{lem}\label{lem:Lipschitz}
    Under Assumption~\ref{assume:f}, $\bar u_t$ and $\bar u_{yy}-\bar u_y$ are Lipschitz, where $\bar u$ is the solution to \eqref{eq:forward-baru}. 
\end{lem}
\begin{proof}
 Let us take a smooth modification $\tilde K$ of $K$ such that 
    \begin{align*}
        &\tilde K(y) = K(y), \qquad \ \ \ \text{ for } y \in [m_0,M_0], \\
       & 0<h < \tilde K'(y) < H \ \,  \, \, \text{ for }  y \in \mathbb R, \\
        &\tilde K'' \geq 0, \qquad \qquad \quad \ \ \text{ for } y \in \mathbb R. 
    \end{align*}
According to Lemma~\ref{lem:PDEbound}, the unique solution to the forward PDE 
\begin{equation*}
\begin{cases}
\bar u_t(t,y)=\tilde K \bigl(\bar u_{yy}(t,y)-\bar u_y(t,y)\bigr),
& (t,y)\in(0,1]\times\mathbb R,\\[0.3em]
\bar u(0,y)=f(y), & y\in\mathbb R,
\end{cases}
\end{equation*}
is exactly given by \eqref{eq:forward-baru}. Setting $\phi=\bar u_t$, it satisfies the quasi linear PDE 
\begin{equation*}
\begin{cases}
\phi_t(t,y)= \tilde K' (\bar u_{yy}(t,y)-\bar u_y(t,y) ) \left(\phi_{yy}(t,y)- \phi_y(t,y)\right)
& (t,y)\in(0,1]\times\mathbb R\\[0.3em]
\phi(0,y)=\tilde K (f''(y)-f'(y)) & y\in\mathbb R,
\end{cases}
\end{equation*}
which is uniform elliptic. Invoking \cite[Section 5.5, Theorem 2]{K}, we get that its solution $\phi=\bar u_t$ is Lipschitz. Restricted to $[m_0,M_0]$, function $\tilde K$ is bi-Lipschitz, and $\bar u_{yy} - \bar u_y=K^{-1}(\bar u_t)$ is also Lipschitz. 
\end{proof}

With the above results at hand, we show that the delta hedging strategy 
\begin{align}\label{eq:deltahedging}
\tilde{\gamma}^n:=\left( v_x\left(\frac{i+1}{n}, S^n_{i}\right)   \right)_{0 \leq i \leq n-1}
\end{align}
is asymptotically optimal for the exponential hedging in \eqref{2.3}. Define its corresponding value 
\begin{align}
    \tilde{\mathcal{C}}_n:= \frac{1}{nl}\log \mathbb E_{\mathbb P} \left[ \exp \left(nl \left(F(S^n_n)-\sum_{i=0}^{n-1}\tilde\gamma^n_i(S^{n}_{i+1}-S^n_i) \right)\right)\right].
\end{align}

\begin{thm}\label{thm:2}
    Under Assumption~\ref{assume:f}, it holds that,
    \begin{align*}
        \lim\limits_{n \to \infty} \mathcal{C}_n= \lim\limits_{n \to\infty} \tilde{\mathcal{C}}_n,
    \end{align*}
    that is, the delta strategies defined in \eqref{eq:deltahedging} is asymptotically optimal. 
\end{thm}

\begin{proof}
 By duality \eqref{2.3} and the HJB \eqref{eq:backward-v}, it suffices to prove 
    \[
         \limsup_{n \to\infty} \tilde{\mathcal{C}}_n \leq v(0,1).
    \]
Let us define for $i=0,\dotso, n-1$
\begin{align*}
&t^n_i:=\frac{i}{n},
\qquad
\Delta S_i^n:=S_{i+1}^n-S_i^n,
\qquad
\tilde\gamma_i^n:=v_x(t^n_{i+1},S_i^n), \\
&I_i^n
:=
v(t^n_{i+1},S_{i+1}^n)-v(t^n_i,S_i^n)-v_x(t^n_{i+1},S_i^n)\Delta S_i^n
,
\end{align*}
Using \(v(1,x)=F(x)\) and the equality
\[
\sum_{i=0}^{n-1} I_i^n
=
v(1,S_n^n)-v(0,S_0^n)-\sum_{i=0}^{n-1}\tilde\gamma_i^n(S_{i+1}^n-S_i^n),
\]
it suffices to prove
\[
\limsup_{n\to\infty}
\frac{1}{nl}
\log
\mathbb E_{\mathbb P}
\left[
\exp\left( nl \sum_{i=0}^{n-1} I_i^n\right)
\right]
\le 0.
\]

By the fundamental theorem of calculus in time,
\begin{align*}
&v(t^n_{i+1},S_{i+1}^n)-v(t^n_i,S_i^n)-v_x(t^n_{i+1},S_i^n)\Delta S_i^n \\
&=
\int_{t_i}^{t_{i+1}} v_t(s,S_i^n)\,ds
+
\int_{S_i^n}^{S_{i+1}^n}(S_{i+1}^n-y)\,v_{xx}(t^n_{i+1},y)\,dy.
\end{align*}
Denote by $A^n_i$ and $B^n_i$ the first and the second term on the right hand side of the equation above.

Since \(v_t\) is Lipschitz in time by Lemma~\ref{lem:Lipschitz}, there exists \(L>0\) such that
\begin{align*}
\left|nA_i^n-v_t(t^n_i,S_i^n)\right|
&\le
n\int_{t^n_i}^{t^n_{i+1}} |v_t(s,S_i^n)-v_t(t^n_i,S_i^n)|\,ds \\
&\le
n\int_{t^n_i}^{t^n_{i+1}} L(s-t^n_i)\,ds
=
\frac{L}{2n}.
\end{align*}
Next, using the identity
\[
\int_x^{x+\delta}(x+\delta-y)f(y)\,dy
=
\delta^2\int_0^1 (1-\theta)f(x+\theta\delta)\,d\theta,
\]
valid for any \(x,\delta\in\mathbb R\), we obtain
\begin{align*}
nB_i^n
&=
n(\Delta S_i^n)^2
\int_0^1 (1-\theta)\,
v_{xx}(t^n_{i+1},S_i^n+\theta\Delta S_i^n)\,d\theta \\
&=
n(S_i^n)^2(h_i^n)^2
\int_0^1 (1-\theta)\,
v_{xx}(t^n_{i+1},S_i^n(1+\theta h_i^n))\,d\theta,
\end{align*}
where $h^n_i:=\frac{\bar \sigma}{\sqrt{n}}\xi_{i+1}$.
Denoting
\[
\Gamma(t,x)=x^2v_{xx}(t,x),
\]
it follows that
\[
n B_i^n
=
\bar\sigma^2\xi_{i+1}^2
\int_0^1 (1-\theta)\,
\frac{\Gamma(t^n_{i+1},S_i^n(1+\theta h_i^n))}{(1+\theta h_i^n)^2}\,d\theta.
\]

Noting that $\Gamma(t,x)= \bar u_{yy}(1-t, \log x) -\bar u_y (1-t, \log x)$ is Lipschitz in $(t,\log x)$ by Lemma~\ref{lem:Lipschitz}, we have the inequality
\[
\left|
\frac{\Gamma(t^n_{i+1},S_i^n(1+\theta h_i^n))}{(1+\theta h_i^n)^2}
-
\Gamma(t^n_i,S_i^n)
\right|
\le \frac{C}{\sqrt n},
\]
uniformly in \(i\) and \(\theta\in[0,1]\). Consequently,
\begin{align*}
\left|
nB_i^n-\frac{\bar\sigma^2}{2}\Gamma(t^n_i,S_i^n)\xi_{i+1}^2
\right| \le \frac{C}{\sqrt n}.
\end{align*}

Combining the bounds for \(A_i^n\) and \(B_i^n\), we obtain
\[
n I_i^n
=
v_t(t^n_i,S_i^n)
+\frac{\bar\sigma^2}{2}\Gamma(t^n_i,S_i^n)\xi_{i+1}^2
+R_i^n,
\]
with
\(
|R_i^n|\le \frac{C}{\sqrt n}.
\) Therefore, we obtain
\begin{align*}
\mathbb E_{\mathbb P}\left[\exp(n l I_i^n)\mid\mathcal F_i\right]
&\le
\exp\left(l v_t(t^n_i,S_i^n)+\frac{C}{\sqrt n}\right)
\mathbb E_{\mathbb P}\left[
\exp\left(\frac{\bar\sigma^2 l}{2}\Gamma(t^n_i,S_i^n)\xi_{i+1}^2\right)
\Bigm|\mathcal F_i
\right] \\
&=\exp\left(l v_t(t^n_i,S_i^n)+ l K(\Gamma(t^n_i,S^n_i))+\frac{C}{\sqrt n}\right),
\end{align*}
where we used the condition law of $\xi_{i+1}$ under $\mathbb P$.

Since \(v\) solves the HJB equation \eqref{eq:backward-v}
\[
v_t(t,x)+K(\Gamma(t,x))=0,
\]
it follows that
\begin{align*}
\mathbb E_{\mathbb P}\left[\exp\left( nl \sum_{i=0}^{n-1} I_i^n\right)\right]
&=
\mathbb E_{\mathbb P}\left[
\exp\left( nl \sum_{i=0}^{n-2} I_i^n\right)
\mathbb E_{\mathbb P}\left[e^{nl I_{n-1}^n}\mid\mathcal F_{n-1}\right]
\right] \\
&\le
e^{C/\sqrt n}
\mathbb E_{\mathbb P} \left[\exp\left( nl \sum_{i=0}^{n-2}  I_i^n\right) \right] \\
&\le \cdots \le e^{Cn/\sqrt n}=e^{C\sqrt n}.
\end{align*}
Therefore we conclude that 
\[
\limsup_{n\to\infty}
\frac{1}{nl}
\log
\mathbb E_{\mathbb P}
\left[
\exp\left( nl \sum_{i=0}^{n-1} I_i^n\right)
\right]
\le 0.
\]
\end{proof}

\begin{example}
    If $F(x)= \alpha+\beta \log(x)$ with some constant $\alpha,\beta$, then we have an explicit solution to \eqref{eq:backward-v}, 
    $$ v(t,x)= \alpha+\beta \log(x)+\frac{1-t}{l} \log \left((1-p)+p \exp\left(-l\bar\sigma^2 \beta /2 \right) \right).$$
    The delta-hedging strategy $v_x(t,x)=\frac{\beta}{x}$ is time-independent. 
\end{example}

\end{document}